\documentclass[preprint,12pt,authoryear]{elsarticle}

\usepackage{amssymb}
\usepackage{setspace} 
\usepackage[english]{babel}
\usepackage[margin=1in]{geometry}
\usepackage[OT1]{fontenc}
\usepackage{graphicx}
\usepackage{algorithmic}
\usepackage{rotating,color}
\usepackage{amsmath, amsthm, amsfonts}
\usepackage{bm}
\usepackage[dvipsnames]{xcolor}
\allowdisplaybreaks

\newtheorem{proposition}{Proposition}[]

\newtheorem{definition}{Definition}[]

\journal{Journal of Multivariate Analysis}

\begin{document}

\begin{frontmatter}



\title{A novel framework for joint sparse \\ clustering and alignment of functional data}


\author{Valeria Vitelli}

\address{Oslo Center for Biostatistics and Epidemiology, Department of Biostatistics, \\ University of Oslo, Sognsvannsveien 9, 0372 Oslo, Norway}

\begin{abstract}
We propose a novel framework for sparse functional clustering that also embeds an alignment step. Sparse functional clustering means finding a grouping structure while jointly detecting the parts of the curves' domains where their grouping structure shows the most. Misalignment is a well-known issue in functional data analysis, that can heavily affect functional clustering results if not properly handled. Therefore, we develop a sparse functional clustering procedure that accounts for the possible curve misalignment: the coherence of the functional measure used in the clustering step to the class where the warping functions are chosen is ensured, and the well-posedness of the sparse clustering problem is proved. A possible implementing algorithm is also proposed, that jointly performs all these tasks: clustering, alignment, and domain selection. The method is tested on simulated data in various realistic situations, and its application to the Berkeley Growth Study data and to the AneuRisk65 data set is discussed.
\end{abstract}




\end{frontmatter}


\section{Introduction}\label{sec:intro}

Finding sparse solutions to clustering problems has emerged as a hot topic in recent years, both in the statistics and in the machine learning community (see \cite{fm}, \cite{twitt} and \cite{elhamifar2013sparse}, among the many possible references). However, only very recently the problem started to receive attention also in the functional data community, surprisingly enough given the natural link between the two frameworks: functional data analysis (FDA) aims at developing statistical methods for infinite-dimensional data, and sparsity is more likely when data are high-dimensional. There are two main research lines where sparsity has been treated in the FDA literature: (i) as each functional datum is in reality observed on a finite grid, sparsity may occur through a local heterogeneity in the frequency of the discretization, or in other words the data might be \emph{sparsely observed} \citep{james2000principal, yao2005functional, muller2010dynamic, di2014multilevel, yao2016probability, stefanucci2018pca, kraus2019inferential}; (ii) more naturally linked to the functional form of the data, sparsity can be defined as local heterogeneity in the concentration of information into the functional space: this aspect has been studied in the regression setting, mainly under linear assumptions \citep{james2009functional, mckeague2010fractals, kneip2011factor, aneiros2014variable, novo2019fast}. There has also been some work on sparse functional PCA \citep{allen2013sparse, huang2009analysis, di2009multilevel}, and on domain selection in inferential approaches such as functional ANOVA \citep{pini2016interval, pini2017interval, pini2018domain}. However, very few proposals deal with unsupervised learning: when it comes to methods specifically targeted to clustering, some very recent approaches \citep{fraiman2016feature, berrendero2016shape, berrendero2016variable, fv17} propose to cluster the curves while jointly detecting the most relevant portions of their domain to the clustering purposes. The latent assumption is redundancy of information: curves cluster only on a reduced part of the domain, while most of the observed functional signal might be totally unrelevant to the scopes of the analysis. 

\begin{figure}[t]
\centerline{\includegraphics[width=\textwidth]{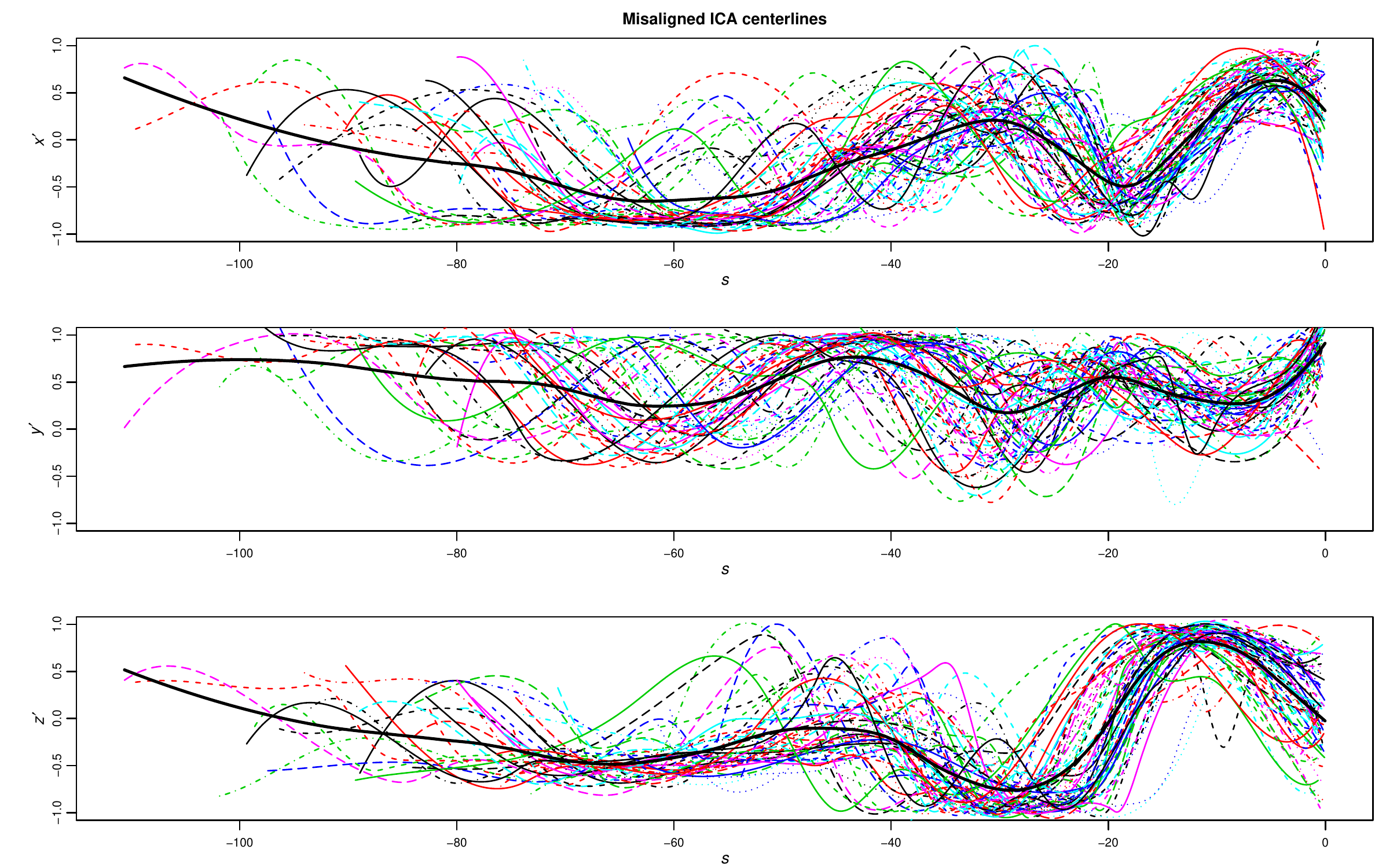}}
\caption{data from the AneuRisk65 data set. From top to bottom, first derivative of the three spatial coordinates ($x^\prime$, $y^\prime$ and $z^\prime$, respectively) of the ICA centerlines of 65 subjects suspected to be affected by cerebral aneurysm. \label{fig:aneurisk_data}}
\end{figure}

A frequent issue when dealing with curve clustering is misalignment: curves might show a {\em phase variability}, describing some abscissa variation ancillary to the scopes of the analysis (e.g., biological clock, measurement start, length of the task,$\ldots$), while the scope of a clustering method is, instead, to capture the {\em amplitude variability}. An evident example of this problem is shown in Figure \ref{fig:aneurisk_data}, where the curves from the AneuRisk65\footnote{The AneuRisk project is a a scientific endeavour that aimed at investigating the role of vessel morphology, blood fluid dynamics and biomechanical properties of the vascular wall, on the pathogenesis of cerebral aneurysms. The project has gathered together researchers of different scientific fields, ranging from neurosurgery and neuroradiology to statistics, numerical analysis and bio-engineering (see \texttt{https://statistics.mox.polimi.it/aneurisk/} for more details).} dataset are displayed: the data consist in the three spatial coordinates (in mm) of 65 Internal Carotid Artery (ICA) centerlines, measured on a fine grid of points along a curvilinear abscissa (in mm), decreasing from the terminal bifurcation of the ICA towards the heart. Estimates of these three-dimensional curves are obtained by means of three-dimensional free-knot regression splines, as described in \cite{sangalli2009efficient}. The first derivatives $x^\prime$, $y^\prime$ and $z^\prime$ (preferred to the ICA original coordinates, which depend on the location of the scanned volume) of the estimated curves are displayed in the figure: the apparent phase variability is strongly dependent on the dimensions and proportions of patients' skulls, and is a nuisance variability meaningless when the study of the heterogeneity of the cerebral morphology is concerned.

Misalignment is a frequent problem, which might heavily affect and confound the clustering results if not properly handled \citep{marron2015functional}: many authors have studied this problem in recent years (see the nice review by \cite{jacques2014functional}), and claim that jointly decoupling the two sources of variation is often convenient \citep{chudova2002probabilistic,gaffney2005joint,ssvv,mattar2012unsupervised,zhang2014joint}. When aiming at sparse clustering, and when the curves are also misaligned, a very common approach is to first align the curves and then use a sparse functional clustering method to estimate the groups and select the domain (as we have previously proposed for analysing the Berkeley Growth study data in \cite{fv17}). However, it is well-known that aligning and clustering the curves jointly is beneficial for the analysis, and many methods to jointly cluster and align curves have already been proposed in the literature on functional data. The starting point for this paper is thus to adapt functional sparse clustering to align the data as well. The only available alternative to such purpose is the method proposed in \cite{cremona2018probabilistic} for the analysis of omics signals, that identifies local amplitude features (the so called ``motifs'') while accounting for phase variability. However, the latter method is targeted to the idea of discovering the functional motifs (i.e., typical local shapes that might recur within each curve, or across several curves in the sample), while our framework follows the ``domain selection'' strategy.

In this paper we propose a functional clustering approach that jointly aligns the curves, and that also performs domain selection, meaning that it estimates the portions of the domain where the grouping structure mostly shows. The problem is correctly framed, from a mathematical viewpoint, as a variational problem having the grouping structure, the warping functions, and the domain selector as a solution. Some properties are required to the functional space (and to the associated metric in particular), and to the class where the warping functions are chosen, to ensure the well-posedness of such variational problem. An iterative implementing algorithm, rooted both on the $K$-means clustering and alignment algorithm introduced in \cite{ssvv}, and on the sparse functional clustering method described in \cite{fv17}, is also proposed. The paper is structured as follows: Section \ref{sec:method} gives all the methodological details on our sparse clustering and alignment proposal, together with the associated implementing algorithm. The results of two simulation studies are reported in Section \ref{sec:simulations}, and the application of the proposed methodology to the Berkeley Growth Study data \citep{tuddenham1954physical} and to the AneuRisk 65 dataset \citep{sangalli2009case} are described in Sections \ref{sec:growth} and \ref{sec:aneurisk}, respectively. Finally, some conclusions and discussion of further developments are given in Section \ref{sec:conc}.

\section{Joint sparse clustering and alignment: a unified framework}\label{sec:method}

In this section we provide the general framework for embedding both the alignment step and the sparsity constraint into the variational problem describing curve clustering. The clustering and alignment variational problem is first defined in subsection \ref{subsec:align}, then the sparsity constraint is described in subsection \ref{subsec:sparse}, and finally the unified variational problem for solving the sparse clustering and alignment problem is proposed in subsection \ref{subsec:all}. Some theoretical considerations on consistency and well-posedness in a more general mathematical setting are also discussed in subsection \ref{subsec:space}.

\subsection{Joint clustering and alignment: decoupling phase and amplitude variability}\label{subsec:align}

Let $(D,\mathcal{F},\mu)$ be a measure space, with $D \subset \mathbb{R}$ compact, $\mathcal{F}$ a functional space suited to the application at hand (assume $L^2(D)$ for the time being), and $\mu$ the Lebesgue measure. Let $f_1,\ldots,f_n \in \mathcal{F}$ be a functional dataset, and $W$ the class of warping functions indicating the allowed transformations of the abscissa to align the curves. 

A very popular approach in the FDA literature is to study the amplitude and phase variation in a functional dataset via equivalence classes (see \cite{srivastava2011registration} and \cite{srivastava2016functional} for a nice introduction): in such approach, clustering can be defined via a proper and mathematically consistent optimization problem, where alignment is performed by selecting, for each curve in the dataset, an optimal aligning function within a class $W$. To ensure the well-posedness of the resulting optimization problem, some properties are required to $\mathcal{F}$ and $W$ \citep{vantini2012definition}:
\begin{itemize}
\item[(i)] $\mathcal{F}$ is a metric space equipped with the metric $d(\cdot,\cdot)\!:\mathcal{F}\times\mathcal{F}\to \mathbb{R}^+_0$ that measures the distance between two curves;
\item[(ii)] $W$ is a sub-group, with respect to composition $\circ$, of the group of continuous automorphisms;
\item[(iii)] $\forall f \in \mathcal{F}, \forall h \in W$ we have that $f \circ h \in \mathcal{F};$
\item[(iv)] \emph{W-invariance} property of $d(\cdot,\cdot)$: for any $f_1,f_2\in \mathcal{F}$ and any $h \in W$
\begin{equation}\label{eq:Winvariance}
d(f_1,f_2)=d(f_1 \circ h, f_2 \circ h),
\end{equation}
meaning that it is not possible to create fictitious increments in the distance between two curves $f_1$ and $f_2$ only by warping both $f_1$ and $f_2$ with $h$.
\end{itemize}
While properties (i)-(iii) are quite natural requirements, the property (iv) of W-invariance of the metric $d$ as expressed in (\ref{eq:Winvariance}) is particularly crucial, since it ensures that we can compare equivalence classes of functions, and not just individual functions: specifically, the equivalence class of a function $f\in\mathcal{F}$ includes all the versions of the function after its warping via an element of $h \in W$, i.e., $f \circ h.$ Therefore, the properties (i)-(iv) above ensure a coherent framework for decoupling phase and amplitude variability: phase variation is incorporated within equivalence classes, while amplitude variation appears across equivalence classes. More precisely:
\begin{definition}
Aligning $f_1\in \mathcal{F}$ to $f_2\in\mathcal{F}$ according to $(d,W)$ means finding $h^* \in W$ that minimizes $d(f_1\!\circ\!h,f_2)$ with respect to $h$.
\end{definition}
\noindent Hence, phase variability is captured by the optimal warping function $h^*$, while amplitude variability is the remaining variability between $f_1\!\circ\!h^*$ and $f_2$ (see \cite{vantini2012definition} for a very nice description of this framework).

When the aim is clustering and aligning the functional dataset $f_1,\ldots,f_n \in \mathcal{F}$, the most natural approach within this framework is to focus on distance-based approaches, such as $K$-means: these methods entail finding the optimal cluster partition $\{C_1,\ldots,C_K\}$ by solving an optimization problem based on a measure of the between-cluster distances, depending on the chosen metric $d.$ When fixing $K$ groups, and when focussing on $K$-means clustering, the clustering and alignment variational problem can be then given as: find the set of optimal aligning functions $\{\tilde{h}_1, \ldots, \tilde{h}_n\}\in W$ and the best data partition $\{\tilde{C}_1,\ldots,\tilde{C}_K\}$ maximising 
\begin{equation}\label{eq:clustAndAlign}
\max_{\{h_1,\ldots,h_n\}\in W,\{C_1,\ldots,C_K\}} \int_D g(f_1(h_1(x)),\ldots,f_n(h_n(x));\{C_1,\ldots,C_K\})dx,
\end{equation}
where we assume $f_i \in L^2(D_i)$, $D_i$ being the warped domain of each function (compact), and $D = \bigcup_{i=1}^n D_i $ being the union of the warped domains of the functions in the dataset. We use $dx$ instead of $d\mu(x)$ for ease of notation. The function $g(\cdot)$ defines the point-wise between-cluster sum-of-squares of the functions in the dataset, which depends on $d$. A quite natural choice for $d$ in this functional space is the normalized $L^2$ norm
\begin{equation}\label{eq:dL2norm}
d(f_1, f_2) = \frac{1}{\sqrt{\mu(D_1 \cap D_2)}} \left( \int_{D_1 \cap D_2} (f_1(x)-f_2(x))^2 dx \right)^{1/2},
\end{equation}
since this choice ensures that properties (i)-(iv) above hold when the group of strictly increasing affine transformations is used as the class for warping functions
\begin{equation}\label{eq:Waffine}
W = \{h : h(x) = ax + b \text{ with } a \in \mathbb{R}^+, b \in \mathbb{R}\}.
\end{equation}
Given the choice in (\ref{eq:dL2norm}), the point-wise between-cluster sum-of-squares in (\ref{eq:clustAndAlign}) takes the form
\begin{align}\label{eq:BCSS}
& g\left( f_1(h_1(x)),\ldots,f_n(h_n(x));\{C_1,\ldots,C_K\} \right) = \nonumber\\ 
& = \frac{1}{n} \sum^n_{i,j = 1} \left\{ \left( f_i(h_i(x)) - f_j(h_j(x)) \right)^2 \cdot \frac{1_{\{h_i^{-1}(D_i) \cap h_j^{-1}(D_j)\}}(x)}{\sqrt{\mu(h_i^{-1}(D_i) \cap h_j^{-1}(D_j))}} \right\} + \nonumber\\
&  - \sum^K_{k=1} \frac{1}{n_k} \sum_{i,j \in C_k} \left\{ \left( f_i(h_i(x)) - f_j(h_j(x)) \right)^2 \cdot \frac{1_{\{h_i^{-1}(D_i) \cap h_j^{-1}(D_j)\}}(x)}{\sqrt{\mu(h_i^{-1}(D_i) \cap h_j^{-1}(D_j))}} \right\},
\end{align}
where $n_k = |C_k|$ for all $k=1,\ldots, K.$

The strategy proposed in \cite{ssvv} for tacking the variational problem (\ref{eq:clustAndAlign}) is to use an iterative algorithm alternating between the following two steps, the so called \emph{K-means alignment} (KMA):
\begin{enumerate}
\item holding the cluster partition $\{C_1,\ldots,C_K\}$ fixed, solve the alignment problem alone by maximizing the criterion to find the best set of warping functions $\{\tilde{h}_1,\ldots,\tilde{h}_n\}$;
\item holding the set of warping functions $\{\tilde{h}_1,\ldots,\tilde{h}_n\}$ fixed, find the best partition $\left\{\tilde{C}_1,\ldots,\right.$ $\left.\tilde{C}_K\right\}$ via a $K$-means step applied to the aligned function.
\end{enumerate}
\cite{ssvv} suggest to add to step 1 above a normalization step, so that the average warping undergone by the curves within each cluster is the identity transformation $h(x)=x$. Normalization is used to select, among all candidate solutions to the optimization problem, the one that leaves the average locations of the clusters unchanged, thus avoiding the drifting apart of clusters or the global drifting of the overall set of curves. We keep the same normalization to the purposes of the present paper, and for $k=1,\ldots, K$ define
\begin{equation}\label{eq:normalization}
\bar{h}_k(x) = \frac{1}{n_k}\sum_{i \in C_k} \tilde{h}_i(x),
\end{equation} 
so that the final version of the aligned curves for $i=1,\ldots,n$ is $\tilde{f}_i = f_i \circ \tilde{h}_i \circ (\bar{h}_k)^{-1},$ where $k$ is such that $i \in C_k$.

Note that the variational problem defined by (\ref{eq:clustAndAlign}) and (\ref{eq:BCSS}) is completely general, and can be adapted to the choices of $\mathcal{F},$ $W$ and $d$ most suited to the problem at hand, as detailed in subsection \ref{subsec:space}.

\subsection{The sparsity constraint}\label{subsec:sparse}

The strategy proposed in (\ref{eq:clustAndAlign}) is a mathematically coherent approach for decoupling phase and amplitude variability, and it has proved to work in a variety of situations, also for functional spaces different from $L^2$ \citep{ssvv}. However, the curves might cluster on a reduced part of the domain, even when needing to be aligned. Nonetheless, a clustering might be present in the phase as well, thus confounding the amplitude variation especially when limited to some portions of the domain. Given all these motivations, we would like to embed into the variational problem defined by (\ref{eq:clustAndAlign}) a domain selection step, for clustering and aligning the curves while jointly detecting the most relevant portions of their domain to the clustering purposes. 

Following the approach first described in \cite{fv17}, we will introduce a weighting function $w: D\rightarrow\mathbb{R}$ having the scope of giving weight to those parts of the domain where the clustering shows the most. Two elements are key to this purpose: (a) the variational problem has to be written such that the optimal $w(\cdot)$ adapts to the point-wise between-cluster sum-of-squares function, as a measure of the point-wise level of clusterisation, and (b) a constraint on $w(\cdot)$ must be included to ensure the sparsity of the solution, with a sparsity parameter for tuning the desired sparsity level. Having those two elements in mind, the variational problem defining sparse functional clustering is written as \citep{fv17}
\begin{align}\label{eq:sparseClust}
 \max_{w \in L^2(D); C_1, \ldots  C_K} & \int_{D} w(x) g(f_1(x),\ldots,f_n(x);C_1,\ldots,C_K) dx, \\
 \text{subject to}\ \ & \Vert w(x) \Vert_{L^2(D)} \leq 1,\ \ w(x) \geq 0\  \mu\mbox{-a.e. and}\ \mu(\left\lbrace x \in D: w(x)=0\right\rbrace ) \geq m. \nonumber
\end{align}
The first two constraints on the weighting function ensure the well-posedness of the optimization problem, while the latter is the sparsity constraint: $w(\cdot)$ must be zero on a set of measure at least $m$. This constraint has the practical implication that, in order to reach the global maximum, those portions of the domain where the point-wise between-cluster sum-of-squares function $g(\cdot)$ is minimal have to be given null weight; viceversa, if curves belonging to different clusters differ greatly in a Borel set $B \subset D$, then $w(\cdot)$ should be strictly positive on that subset, with each value \(w(x)\) reflecting the importance of $x\in D$ for partitioning the data. 

When fixing the partition $\{ C_1, \ldots  C_K \},$ a result given in \cite{fv17} provides the explicit form of the unique optimal $w(\cdot)$ for problem (\ref{eq:sparseClust}). However, that result is specific to the choice of the functional space ($L^2$) and the distance (the natural metric in that space). Hence, aim of the remaining part of the section is to prove the validity of the result when combined with alignment (subsection \ref{subsec:all}), and then proving its generality with respect to the choice of the metric and functional space (subsection \ref{subsec:space}).

\vspace{.1cm}
\noindent \emph{Remark.} \cite{fv17} have already investigated the possibility of changing the sparsity constraint to a more classical $L^1$ penalization, for achieving the same weighting effect in a soft-thresholding fashion. However, they concluded that this could lead to sub-optimal solutions when compared to the hard-thresholding constraint proposed here, without removing the need for parameter tuning (instead of the sparsity parameter $m$, we would need a new parameter giving the constraint on the $L^1$ norm of $w$). Therefore, we here consider only hard-thresholding to the purposes of combining sparse clustering to an alignment procedure, and leave the possibility of using soft-thresholding to future speculation.

\subsection{The variational problem for joint sparse clustering and alignment}\label{subsec:all}

We here combine the two variational problems given in (\ref{eq:clustAndAlign}) and (\ref{eq:sparseClust}) to provide a unified solution for handling both misalignment and sparse clustering. We define \emph{joint sparse clustering and alignment} as the solution to a variational constrained optimization problem of the form 
\begin{align}\label{eq:sparseClustAlign}
 \max_{w \in L^2(D), \{h_1,\ldots,h_n\}\in W,\{C_1,\ldots,C_K\}} & \int_{D} w(x) g\left( f_1(h_1(x)),\ldots,f_n(h_n(x));\{C_1,\ldots,C_K\} \right) dx, \\
\text{subject to }\Vert w(x) \Vert_{L^2(D)} \leq 1, & w(x)\geq 0\;\mu-\text{a.e.},\;\mu(\left\lbrace x \in D: w(x)=0\right\rbrace ) \geq m. \nonumber
\end{align}
where the point-wise between-cluster sum-of-squares $g(\cdot)$ in (\ref{eq:sparseClustAlign}) takes the form in (\ref{eq:BCSS}). Problem (\ref{eq:sparseClustAlign}) can be tackled via an iterative algorithm that alternates the following steps:
\begin{enumerate}
\item holding $w \in L^2(D)$ and the cluster partition $\{C_1,\ldots,C_K\}$ fixed, solve the alignment problem alone by maximizing the criterion (\ref{eq:sparseClustAlign}) to find the best set of warping functions $\{\tilde{h}_1,\ldots,\tilde{h}_n\}$;
\item holding $w \in L^2(D)$ and the set of warping functions $\{\tilde{h}_1,\ldots,\tilde{h}_n\}$ fixed, find the best partition $\{\tilde{C}_1,\ldots,\tilde{C}_K\}$ via a functional clustering step applied to the aligned functions, using the functional metric $d$ weighted via $w$;
\item holding the set of warping functions $\{\tilde{h}_1,\ldots,\tilde{h}_n\}$ and the partition $\{\tilde{C}_1,\ldots,\tilde{C}_K\}$ fixed, find the optimal weighting function $\tilde{w}(x).$
\end{enumerate}
This approach allows to solve one maximization problem at time, and is proved to increase the objective function at every step. When thinking to the well-posedness of the variational problem, only the first one among the three steps above needs closer inspection. Point 2 is trivial, given that the number of possible partitions is finite, and thus the problem is obviously well-posed and has a unique solution. Point 3 refers to solving a variational problem with the same characteristics as in \cite{fv17}, and can thus be solved using Theorem 2 in that paper, provided that the correct expression of the point-wise between-cluster sum-of-squares is used. Let us thus focus on point 1.

We need to tackle the following variational problem
\begin{equation}\label{eq:align}
\max_{\{h_1,\ldots,h_n\}\in W} \int_D w(x)\left\{ \frac{1}{n} \sum^n_{i,j = 1} G_{ij}(x) - \sum^K_{k=1} \frac{1}{n_k} \sum_{i,j \in C_k} G_{ij}(x) \right\} dx,
\end{equation}
where, from (\ref{eq:BCSS})
$$
G_{ij}(x) = \left( f_i(h_i(x)) - f_j(h_j(x)) \right)^2 \cdot \frac{1_{\{h_i^{-1}(D_i) \cap h_j^{-1}(D_j)\}}(x)}{\sqrt{\mu(h_i^{-1}(D_i) \cap h_j^{-1}(D_j))}}.
$$
Now, given that all sums are finite, we can obviously rearrange problem (\ref{eq:align}) as follows
\begin{equation}\label{eq:align2}
\max_{\{h_1,\ldots,h_n\}\in W} \left\{ \frac{1}{n} \sum^n_{i,j = 1} \int_D w(x)G_{ij}(x)dx - \sum^K_{k=1} \frac{1}{n_k} \sum_{i,j \in C_k} \int_D w(x)G_{ij}(x)dx \right\},
\end{equation}
which is exactly equal to the maximization of the between-cluster sum-of-squares, when using the following weighted distance among functions
\begin{equation}\label{eq:dL2normW}
d_w(f_1, f_2) = \frac{1}{\sqrt{\mu(D_1 \cap D_2)}} \left( \int_{D_1 \cap D_2} w(x)(f_1(x)-f_2(x))^2 dx \right)^{1/2},
\end{equation}
instead of the unweighted norm in (\ref{eq:dL2norm}). The weighted distance defined in (\ref{eq:dL2normW}) enjoys good mathematical properties, among which W-invariance:
\begin{proposition}\label{prop:Winvariance}
The metric $d_w(\cdot, \cdot)$ in (\ref{eq:dL2normW}) is W-invariant when choosing $W$ as in (\ref{eq:Waffine}).
\end{proposition}
\begin{proof}
For proving W-invariance we need to verify property (\ref{eq:Winvariance}). Take $f_1,f_2\in\mathcal{F}$ and $h(x)=ax+b \in W$. Then
\begin{align*}
d_w(f_1 \circ h, f_2 \circ h) & = \frac{\left( \int_{h^{-1}(D_1 \cap D_2)} w(h(x))(f_1(h(x)) - f_2(h(x)))^2 dx\right)^{1/2}}{\sqrt{\mu(h^{-1}(D_1 \cap D_2))}} \\
& = \frac{\left( \int_{D_1 \cap D_2} w(y)(f_1(y) - f_2(y))^2 \frac{dy}{a}\right)^{1/2}}{\sqrt{\int_{h^{-1}(D_1 \cap D_2)}dx}} \\
& = \frac{\left( \int_{D_1 \cap D_2} w(y)(f_1(y) - f_2(y))^2 dy \right)^{1/2} \cdot \frac{1}{\sqrt{a}}}{\sqrt{\int_{D_1 \cap D_2}dy}\cdot\frac{1}{\sqrt{a}}} = d_w(f_1,f_2) 
\end{align*}
where in the second line we used a change of variable defined by $y:=h(x),$ and the fact that in general $\mu(D)=\int_D dx.$
\end{proof}

Therefore, when choosing the functional space $\mathcal{F}=L^2(D),$ the functional norm in (\ref{eq:dL2norm}), and the class of warping functions $W$ of strictly increasing affine transformations, we can conclude that problem (\ref{eq:align}) in step 1 of the algorithm is well-posed, and it correctly decouples phase from amplitude variability. Hence, the variational problem of joint sparse clustering and alignment in (\ref{eq:sparseClustAlign}) is well-posed.

\subsection{Changing the functional space $\mathcal{F}$}\label{subsec:space}

Until now we have assumed the functional space $\mathcal{F}$ for the curves in the dataset to be $L^2.$ Even though this might be a sufficiently good choice in many applications, it is not necessarily always the best, and thus some flexibility with respect to the choice of $\mathcal{F}$ might be beneficial. For instance, the KMA strategy for clustering and alignment as described in subsection \ref{subsec:align} was first introduced in \cite{ssvv} for curves having a continuous first derivative, thus with no lack of generality (and given the natural assumption of compact domains of the curves) we can correctly frame that procedure by choosing $\mathcal{F}$ to be $H^1(D).$ Is it then possible to frame the sparse clustering and alignment setting described above in this space?

The clustering and alignment method proposed in \cite{ssvv} made use of a similarity index among curves $f\in H^1$ (and not of a functional distance) that had nice mathematical properties when used jointly with the class $W$ of warping functions as defined in (\ref{eq:Waffine}). The similarity index was defined as
\begin{equation}\label{eq:rhoKMA}
\rho(f_1,f_2) = \frac{\langle\langle f_1, f_2 \rangle\rangle_{H^1}}{|f_1|_{H^1}|f_2|_{H^1}}
\end{equation}
where for any $f_1,f_2 \in H^1,$ we denote with $\langle\langle \cdot,\cdot \rangle\rangle_{H^1}$ the semi-internal product in $H^1$, i.e.
$$
\langle\langle f,g \rangle\rangle_{H^1} = \langle f^{'}, g^{'} \rangle_{L^2} = \int f^{'} g^{'}
$$
and with $|\cdot|_{H^1}$ the associated semi-norm
$$
|f|_{H^1} = \sqrt{\langle\langle f,f \rangle\rangle_{H^1}} = ||f^{'}||_{L^2} = \sqrt{\int (f^{'})^2}.
$$

The KMA method then performs clustering and alignment via a K-means approach maximising the within-cluster total similarity index as defined by (\ref{eq:rhoKMA}), and choosing the aligning functions in the class $W$ as defined in (\ref{eq:Waffine}). Then, the variational problem for sparse clustering and alignment can still be defined as in (\ref{eq:sparseClustAlign}), provided we choose
\begin{align}\label{eq:WCSSsparseClustAlignRho}
&  g\left( f_1(h_1(x)),\ldots,f_n(h_n(x));\{C_1,\ldots,C_K\} \right)  = \nonumber\\
& = \sum_{k=1}^K\frac{1}{n_k}\sum_{i,j\in C_k}\left\{ \frac{f_i^{'}(h_i(x))}{|f_i|_{H^1}}\cdot\frac{f_j^{'}(h_j(x))}{|f_j|_{H^1}}\cdot 1_{\{h_i^{-1}(D_i) \cap h_j^{-1}(D_j)\}}(x) \right\}.
\end{align}
The algorithm for tackling the variational problem (\ref{eq:sparseClustAlign}) when the $g(\cdot)$ function is based on the within-cluster similarity index as detailed in (\ref{eq:WCSSsparseClustAlignRho}) is unchanged, and the same convergence properties as in the rest of the paper hold (increase of the objective function at every step).

An interesting application of sparse clustering and alignment when using the similarity index introduced in (\ref{eq:rhoKMA}) instead of the distance in $L^2$ is shown in Section \ref{sec:aneurisk}, where the ICA centerlines from the 65 subjects in the AneuRisk65 dataset are analyzed (data are shown in Figure \ref{fig:aneurisk_data}). For these data, as thoroughly discussed in \cite{sangalli2009case}, the use of an $L^2$ type of distance in the variational problem defining clustering and alignment would not make sense, since the functional measure needs to be invariant with respect to curves' dilations or shifts in the 3-D space, which are due to the subjects' different skull dimensions. Hence, the possibility of using a different functional measure in the variational problem (\ref{eq:sparseClustAlign}) defining sparse clustering and alignment is crucial for the analysis of the AneuRisk65 dataset.

\section{Simulations}\label{sec:simulations}

The approach to sparse clustering and alignment of functional data as described in section \ref{subsec:all}, together with its implementing algorithm (shortly named from now on sparse K-mean alignment, or \emph{sparse KMA}) is here tested on synthetic data in various situations. First we test whether the method is sufficiently accurate in an ``easy'' situation (subsection \ref{subsec:sim1}), and then we test it in a more complex setting when we also need to tune $K$, and we compare to the performance of KMA (subsection \ref{subsec:sim2}). In both simulations we use the $L^2$ version of the procedure.

\subsection{Simulation 1: Does sparse clustering and alignment work?}\label{subsec:sim1}
Aim of this simulation study is to test whether sparse KMA works in a relatively ``easy'' setting. The data generation mechanism is as follows:
\begin{itemize}
\item[-] number of classes $K = 2$, both for data generation and for clustering;
\item[-] $100$ observations per class, thus $n = 200$ data in each scenario;
\item[-] cluster true mean functions
\begin{equation}\label{eq:joint1media1g}
f_1(x) = q\cdot x^9 \cdot 1_{[-1,1]}(x),
\end{equation}
\begin{equation}\label{eq:joint1media2}
f_2(x) = q\cdot x^9 \cdot 1_{[-1,0]}(x) + q\cdot x^2 \cdot 1_{[0,1]}(x),
\end{equation}
with $q=1$ (see the top-right panel in Figure \ref{fig:sim1});
\item[-] data $y_i(x)$ $(i=1,\ldots,N)$ are generated from the two mean functions with a random parameter $q_i \sim N(1,0.15^2)$ (see the top-left panel in Figure \ref{fig:sim1});
\item[-] finally, misaligned data are generated from $y_i(x)$ by warping the abscissa $x$ via a curve specific affine warping $h_i(x) = a_i \cdot x + b_i,$ with $a_i \sim U(.9,1.1)$ and $b_i \sim U(-.1,.1)$.
\end{itemize}
The generated data are shown in Figure \ref{fig:sim1}, top-center panel, each curve being coloured according to the true cluster label. It can be noticed that the curves in the 2 clusters are completely overlapped in the negative part of the abscissa, while the difference is evident in the positive part, and it becomes more and more evident towards the domain border. We thus expect to be able to capture this difference via $w(x)$, and we will verify whether the shape of the weighting function also reflects the between-cluster difference.

\begin{figure}[t]
\centering\includegraphics[width=\textwidth]{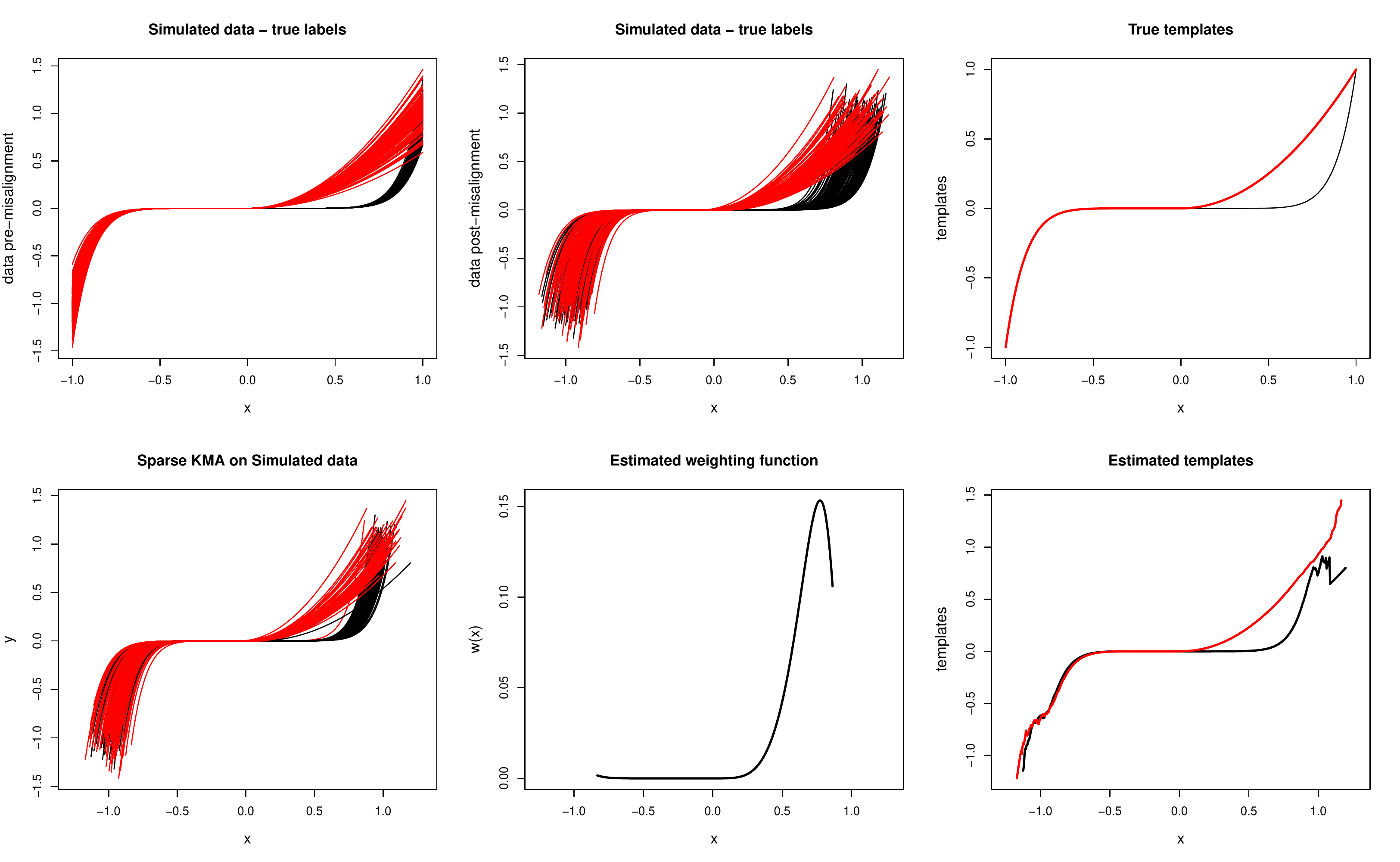}
\caption{one run of sparse KMA on the data of Simulation 1. Top-left, simulated data before misalignment was applied; top-center: simulated curves after misalignment; top-right: true templates used in the simulation; bottom-left, aligned curves after applying sparse KMA; bottom-center, estimated weighting function; bottom-right, estimated cluster templates. In all top panels, curves are coloured according to true labels, while in all bottom panels according to cluster assignments estimated via sparse KMA.\label{fig:sim1}}
\end{figure}

We fix the parameters of the sparse KMA procedure as follows: number of clusters $K=2$ ($K$ is set to the correct number of clusters because we will check the tuning of this parameter in the next section); sparsity parameter $m = 40\%$ (note that one can easily get an idea of this value by simply looking at the curves); tolerance on the functional distance $0.001$ (this is the criterion for stopping the loop, and it refers both to the average functional distance of each datum to the template, and to the average change in $w(\cdot)$, which must be satisfied together with unchanged cluster labels in two subsequent iterations); proportion of allowed warping $.01$ (meaning that we allow, at each iteration in the loop, each $a_i$ to vary in the interval $[.99;1.01]$ and each $b_i$ in the interval $[-.01;.01]$).

The results of one run of sparse KMA are shown in Figure \ref{fig:sim1}. The bottom-left panel shows the aligned curves coloured according to the estimated cluster labels in that run, and the result seems very good, since only 2 curves have been wrongly assigned. The curves alignment, even if satisfactory, does not bring the curves back to the original pre-misalignment shape, as shown in the panel above, and this issue will be further discussed in the next section and in the discussion. The bottom-center panel shows $w(x)$, which correctly points out the relevant portion of the domain (evidently shown in the original data above), and the bottom-right panel shows the cluster templates, which are quite well estimated (if compared to the true templates in the panel above). The average misclassification rate over 20 runs of the procedure was $1.575\%.$ One run of the procedure took on average 35 seconds, and it was converging after 3 to 6 iterations, showing a quite good computational efficiency and stability of the optimization. Initialization of the algorithm was completely random.

\subsection{Simulation 2: Comparison with clustering and alignment, tuning of $K$}\label{subsec:sim2}
The purposes in this second simulation study are several. We first of all would like to test whether sparse KMA works in a more complex setting: more extreme misalignment, and very limited portion of the domain differentiating among clusters. Moreover, we also aim at comparing the results to joint clustering and alignment (KMA) performed \emph{before} sparse clustering (the strategy followed for analysing the growth curves in the Berkeley Growth Study in \cite{fv17}, for instance), i.e., comparing to estimating phase variability first. Finally, we would like to propose a strategy for tuning the number of groups $K$, and to check whether this tuning strategy effectively selects the correct number of groups (for both sparse KMA and KMA).

\begin{figure}[t]
\centering\includegraphics[width=\textwidth]{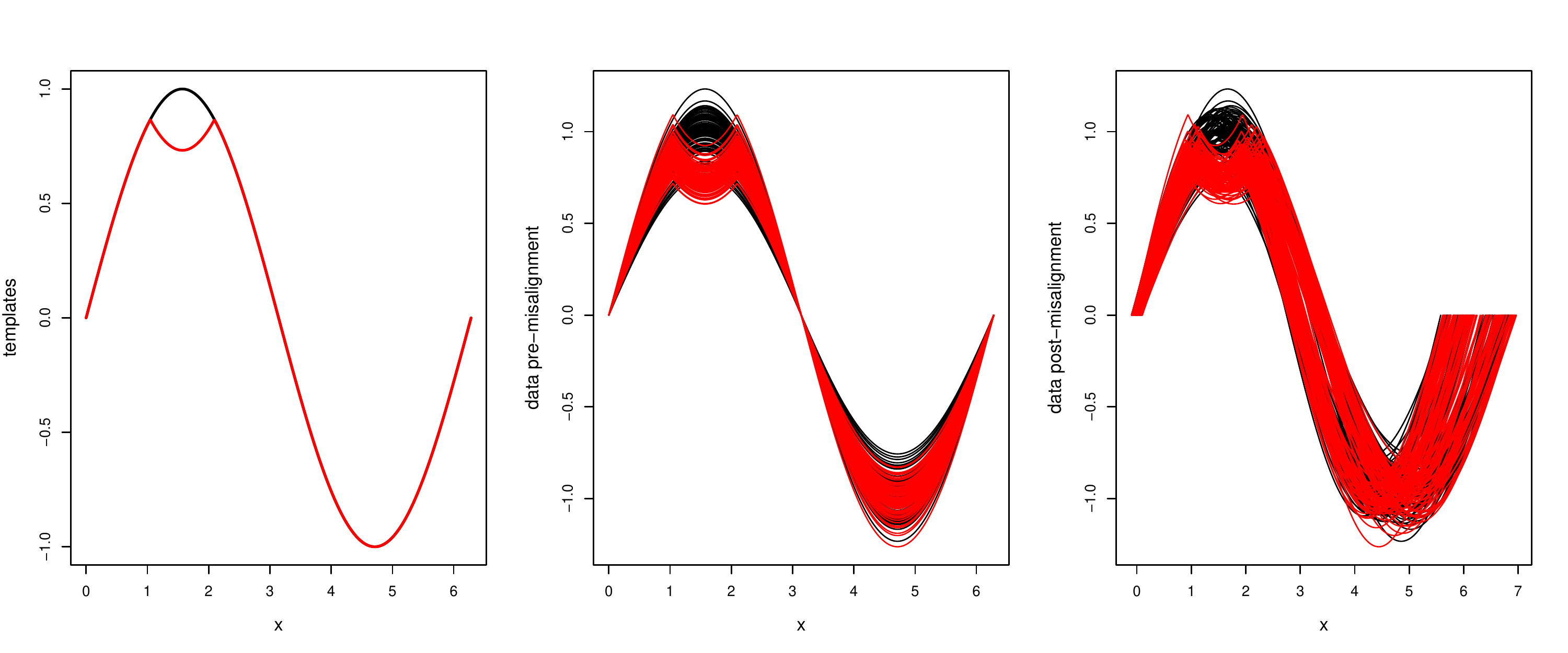}
\caption{one set of data generated for Simulation 2. Left panel, the true cluster templates; center panel, synthetic curves coloured according to true labels, before misalignment; right panel, synthetic curves coloured according to true labels, after misalignment.\label{fig:sim2data}}
\end{figure}

The data generation strategy for this second simulation is the following:
\begin{itemize}
\item[-] $K = 2$ for data generation, and $n = 200$ as before;
\item[-] cluster true mean functions
\begin{equation}\label{eq:joint2media1}
f_1(x) = q \cdot \sin(x) \cdot 1_{[0,2\pi]}(x),
\end{equation}
\begin{equation}\label{eq:joint2media2}
f_2(x) = q \cdot [ \sin(x) \cdot 1_A(x) + (\sqrt 3 - \sin(x)) \cdot 1_{[\pi/3, 2\pi/3]}(x) ],
\end{equation}
with $A=[0,\pi/3]  \cup [2\pi/3, 2\pi]$ (see the left panel in Figure \ref{fig:sim2data});
\item[-] again, data $y_i(x)$ $(i=1,\ldots,N)$ are generated with $q_i \sim N(1,0.15^2)$ (see the center panel in Figure \ref{fig:sim2data});
\item[-] misaligned data are generated by warping the abscissa $x$ via a curve specific affine warping $h_i(x) = a_i \cdot x + b_i;$ we also induce a slight clustering in the phase (visible in the misaligned data, right panel in Figure \ref{fig:sim2data}).
\end{itemize}

The parameters of the sparse KMA procedure are set as follows: tolerance on functional distance $.001$, and proportion of allowed warping $.05.$ The number $K$ of groups in the procedure is varied from 2 to 4, and tuned by looking at the boxplot of the within-cluster distances of the curves to the associated template. Finally, the sparsity parameter $m$ is here fixed to $30\%$ with some trial-and-error, and by inspection of the original data (see the conclusion section for more discussions on this point).
When it comes to performing a comparison of the performance of sparse KMA and KMA, in order for the results of the two methods to be comparable, we need a further specification: the tolerance criterion for stopping the iterative algorithm in sparse KMA must be only on the within-cluster distance, and not on the average change in $w(\cdot).$ Otherwise, the alignment might converge differently for the two procedures. 

\begin{figure}[t]
\centerline{\includegraphics[width=.45\textwidth]{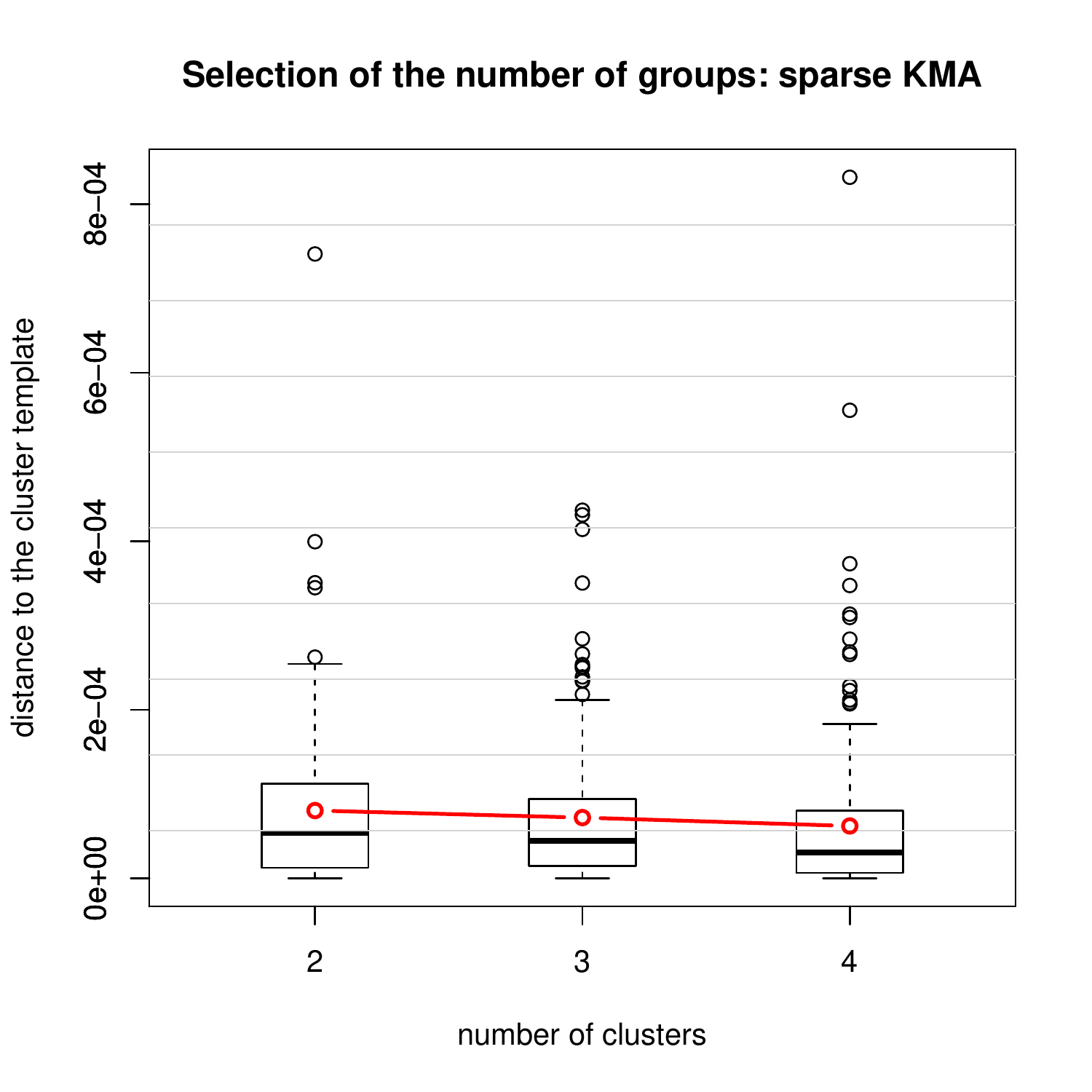}
\includegraphics[width=.45\textwidth]{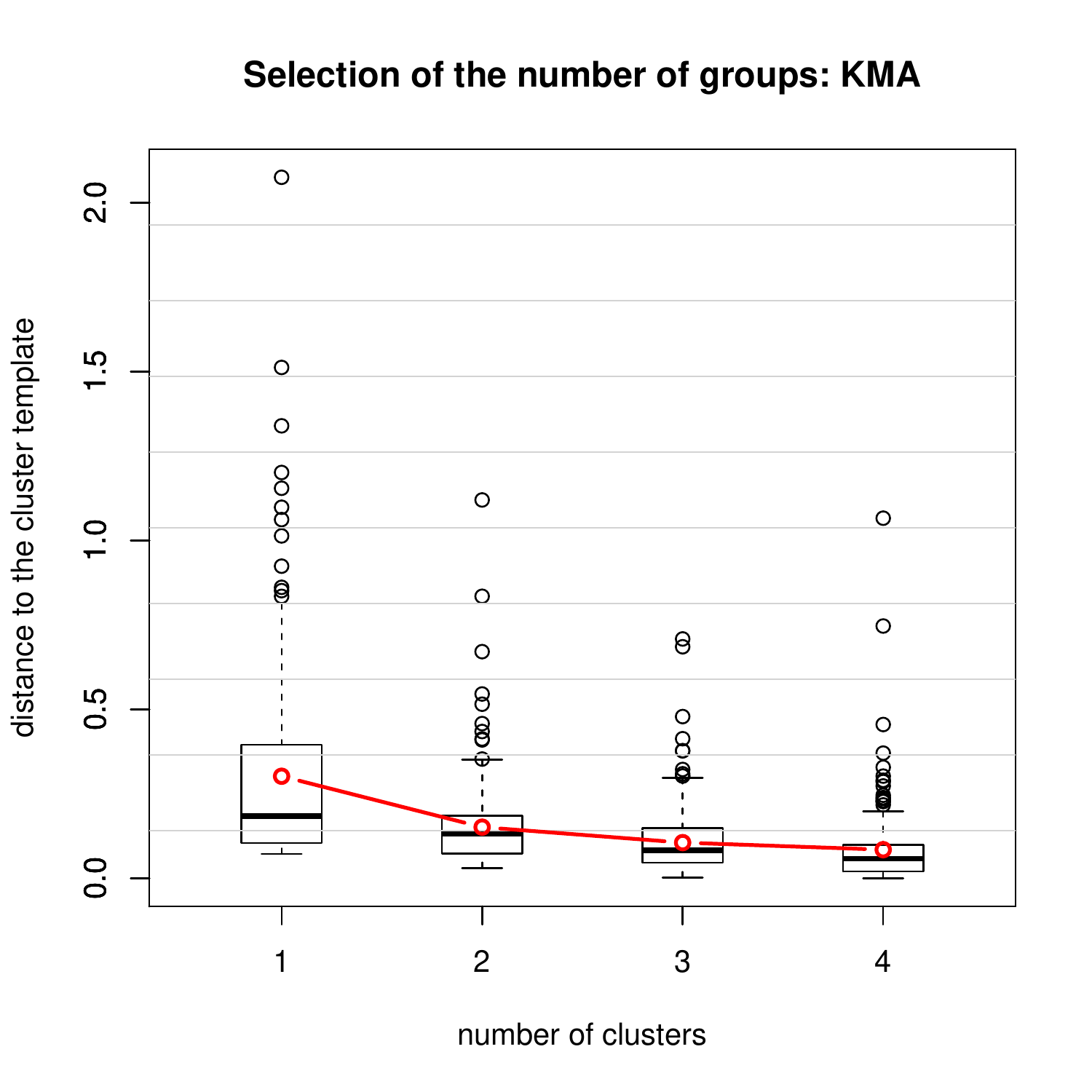}}
\caption{boxplot of the within-cluster distances for one run on the data of Simulation 2, when varying the number of groups. Left panel, sparse KMA; right panel, KMA. The alignment settings are the same for the two procedures.\label{fig:sim2dist}}
\end{figure}

The results obtained via sparse KMA (for one run), in terms of the within-cluster distances, are shown in Figure \ref{fig:sim2dist}, left panel: we see that sparse KMA correctly suggests 2 clusters, since the within-cluster distances do not evidently decrease when increasing the number of groups to more than 2 (p-value of the Mann-Whitney test for significant difference in the within-cluster distances when increasing from 2 to 3 groups was large in all runs). Figure \ref{fig:sim2dist} also shows (right panel) the same within-cluster distances after one run of KMA on the same data, for $K=1,2,3,4$ and with the same parameter settings for alignment as in sparse KMA: the boxplots would suggest 4 clusters (at least) for KMA (p-value 0 for Mann-Whitney in all runs). Note that by definition of KMA, and differently from sparse KMA, the procedure can be run also for $K=1:$ this would be the right conclusion on the number of groups if one aims at performing sparse clustering on the aligned curves afterwards. However, the amplitude variability, even if localized, acts as a confounder for KMA, which fails to align the curves to only one group and would select an overly complex grouping structure.

\begin{figure}[t]
\centerline{\includegraphics[width=\textwidth]{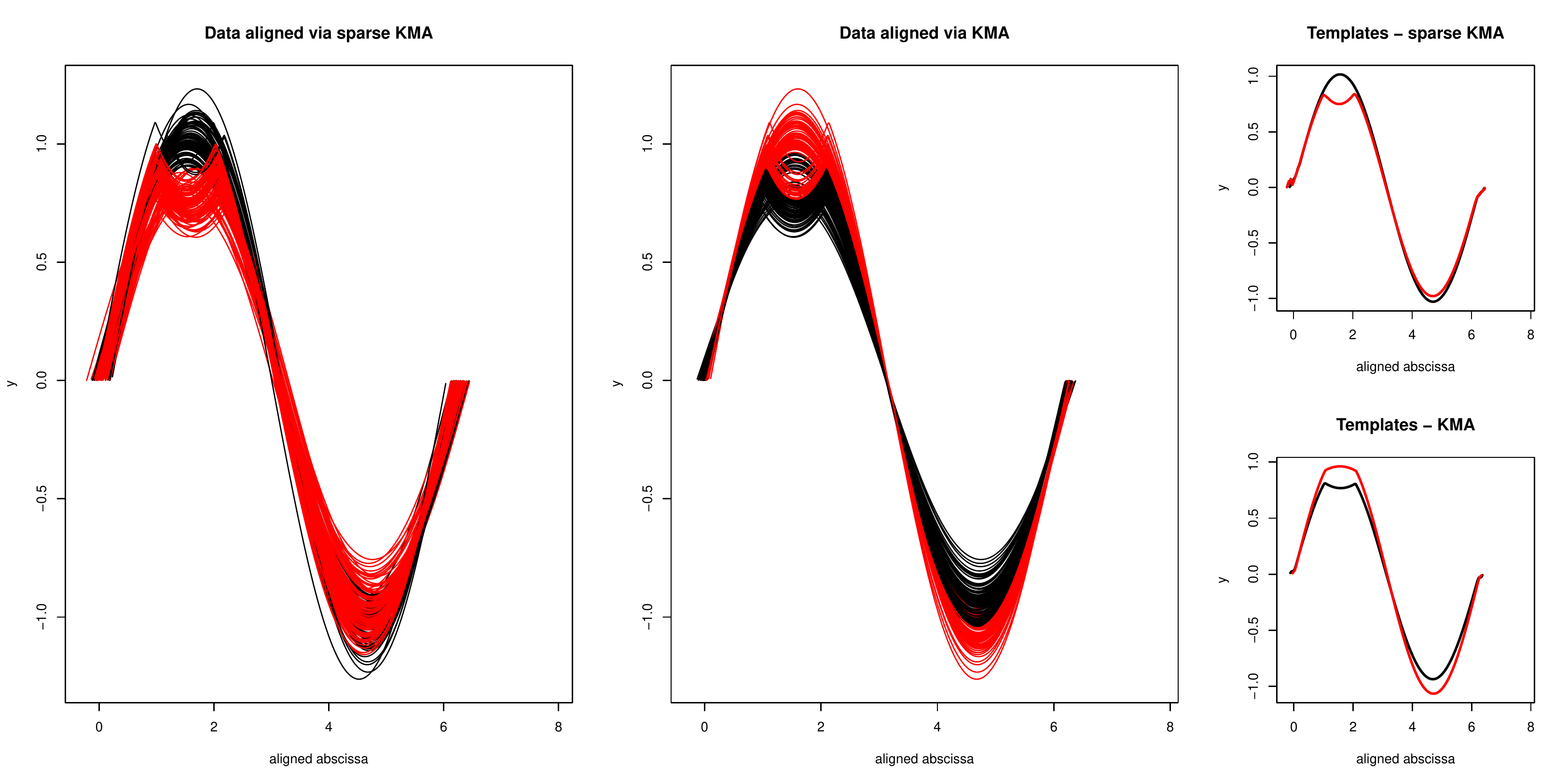}}
\caption{results of one of the runs on the data of Simulation 2, when $K=2$. Left panel, aligned data via sparse KMA, coloured according to the estimated cluster labels; center panel, aligned data via KMA, coloured according to the estimated cluster labels; the 2 right panels show the templates estimated via sparse KMA (top) and KMA (bottom).\label{fig:sim2comparison}}
\end{figure}

Figure \ref{fig:sim2dist} has lead us to conclude that there is an evident difference in performance of sparse KMA when compared to KMA, since the latter procedure seems to fail at selecting the correct number of groups. We will now inspect whether, when fixing $K=2,$ sparse KMA also shows an improved performance in terms of detecting the true groups. First of all, when $K=2$ is used in both procedures, the average misclassification rate over 20 runs of sparse KMA was 9.475\%, while it was nearly 35\% for KMA. Moreover, if we look at the final result of one of the runs, as shown in Figure \ref{fig:sim2comparison}, it is evident that the non-sparse KMA procedure fails at clustering the data because the difference among groups is too much localized in the domain: the data clustered and aligned by sparse KMA (left panel) reflect the true grouping structure, while KMA aligns well but fails at clustering (center panel). This is also reflected in the templates estimated via the two procedures, much closer to the truth for sparse KMA (top-right panel) than for KMA (bottom-right panel).

In the light of the results of Simulation 2, we can conclude that joint sparse clustering and alignment improves quite substantially the clustering results, particularly when the amplitude difference among clusters is localized in the domain. However, the alignment of the data resulting from sparse KMA is slightly worse than with non-sparse clustering and alignment. Conversely, clustering and alignment alone fails to correctly estimate the templates in the relevant portion of the domain, and thus misclassifies quite heavily the data, even though it performs a slightly better alignment. Clustering and alignment alone also overestimates the number of groups, which might be the result of not decoupling correctly phase and amplitude variability in this complex scenario. All in all, we thus conclude that jointly performing sparse clustering and alignment improves the final result of the analysis.

\section{Case Study: the Berkeley Growth Study dataset}\label{sec:growth}

As illustrative example, we consider the Berkeley Growth Study, a very famous benchmark dataset for functional data analysis also provided in the \texttt{fda} package in \textsf{R} \citep{fda-package}. The analysis of these 93 growth curves here reported should be considered as a follow-up to the analysis shown in \cite{fv17}: in that paper, a sparse clustering procedure was used to detect a 2-groups classification on the set of growth velocities, \emph{after} they had been aligned via $1$-mean alignment \citep{ssvv}. The very interesting result of that analysis was that sparse clustering was able to further cluster the data into 2 groups, essentially based on the presence/absence of the mid-spurt in the associated growth velocity curve (see Figure 7 in \cite{fv17}). This finding, together with the results obtained via 1-mean alignment and reported in \cite{sangalli2010functional}, showing a clear gender stratification in the estimated warping functions, would support the following conclusion: the Berkeley Growth Study curves show a gender grouping in the phase variability, a main amplitude feature (the main growth spurt, around 12 years of aligned age) common to all children, and a very localized amplitude clustering according to the presence/absence of the mid-spurt (between 2 and 5 years of aligned age). Aim of the present section is to confirm and validate these interesting findings via a unified procedure able to decouple the localized clustering in the amplitude from the phase variability, and possibly to gain more insight on this thoroughly-studied dataset.

The data in the Berkeley Growth Study consist of the heights (in cm) of 93 children, measured quarterly from 1 to 2 years, annually from 2 to 8 years and biannually from 8 to 18 years. The growth curves and derivatives are estimated, starting from such measurements, by means of monotonic cubic regression splines \citep{rs2}, and the growth velocities are further used in the analysis (see the top-left panel in Figure \ref{fig:growth}). Indeed, when looking at the growth velocities, the growth features shown by most children appear more evident (mainly the main pubertal growth spurt around 12 years), but growth velocities also clearly show that each child follows her own biological clock. Using a sparse functional clustering procedure able to also account for such phase variability would then surely be beneficial.

The sparse KMA procedure is set as follows: we use the $L^2$ setting for coherence to \cite{fv17}, and $K=2$ to be able to compare the current findings with the ones reported in the previous paper. The sparsity parameter $m$ is set to 55\% after some tuning, percentage of allowed warping to 4\%, and tolerance for the iterative procedure to $0.005$, to ensure good convergence. Figure \ref{fig:growth} shows, in the top panels, the original misaligned growth velocities (top-left) and the same growth velocities after sparse 2-mean alignment (top-center), when coloured according to the estimated cluster labels: the resulting clustering structure is essentially based on the presence or absence of the mid-spurt (3-5 years in aligned age), while the onset of the main growth velocity spurt (10-12 years in aligned age) seems mostly aligned across children. This feature distinguishing the clusters can be seen even more evidently in the weighting function (bottom-left panel of Figure \ref{fig:growth}), that clearly points out as relevant part of the domain the mid-spurt onset, and with a much smaller peak the main growth velocity spurt. When inspecting the estimated cluster templates (bottom-center panel of Figure \ref{fig:growth}), the mid-spurt in growth velocity is evidently present in the black cluster template, which shows two growth peaks, and absent in the red cluster template, which is nearly flat before the onset of the main pubertal spurt. 

\begin{figure}[t]
\centerline{\includegraphics[width=\textwidth]{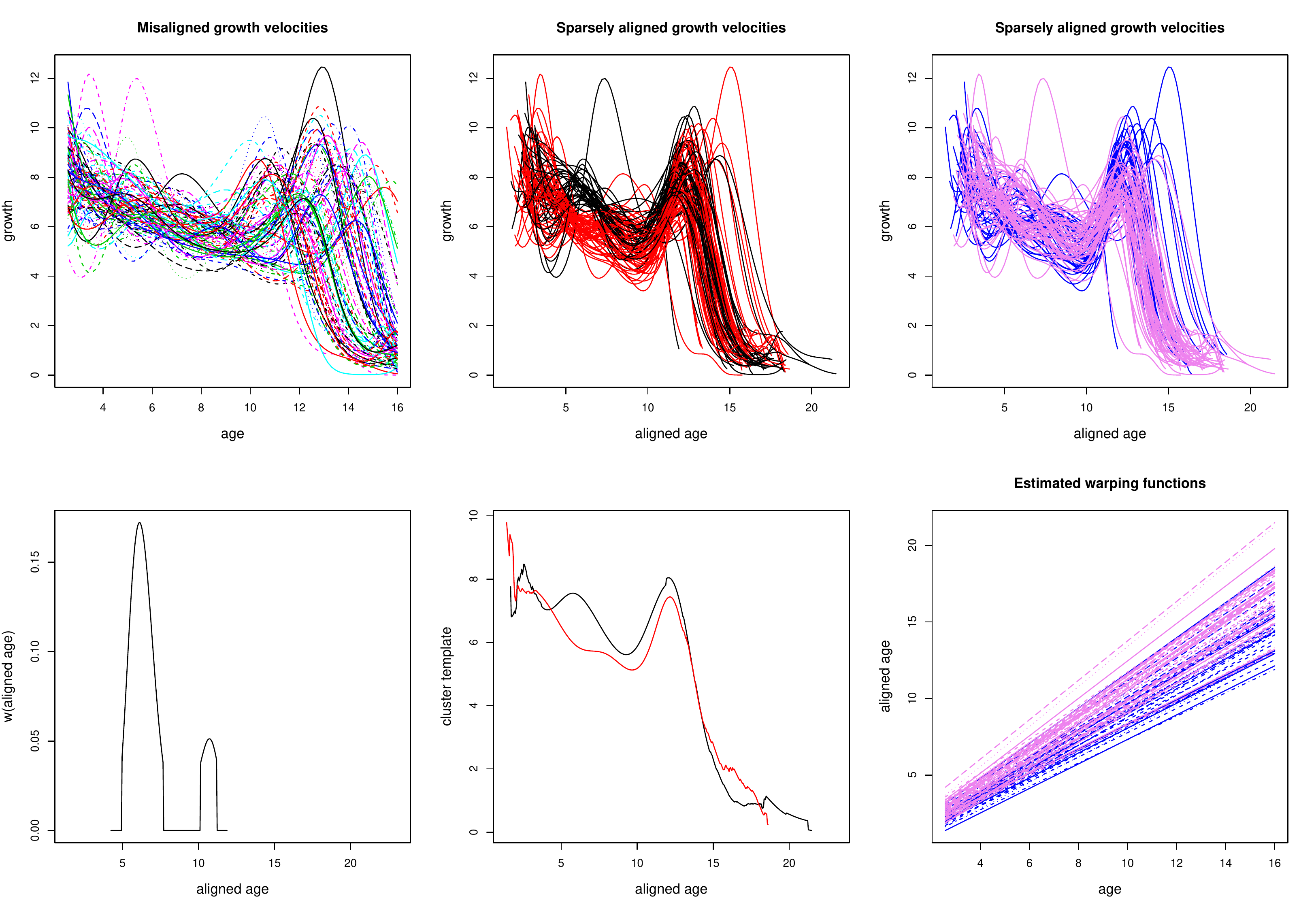}}
\caption{results of sparse KMA on the Berkeley Growth Study data. Top-left, original growth velocities; top-center, clustered and aligned growth velocities via sparse KMA with 2 clusters, coloured according to the estimated cluster labels; top-right, same as top-center but colored in blue for boys and pink for girls; bottom-left, estimated weighting function; bottom-center, estimated templates, with labels colouring (as above); bottom-right, estimated warping functions, with gender colouring (as above).\label{fig:growth}}
\end{figure}

Even if we do not aim at performing any supervised classification of boys and girls, we recognize that gender plays a key role in shaping the kids' growth patterns, mainly in timing their biological clocks. Therefore, we would like to recognize this pattern in the estimated phase variability. The two right-most panels in Figure \ref{fig:growth} explore the relationship between gender stratification and phase/amplitude variability in this dataset: in the top-right panel of Figure \ref{fig:growth}, the same growth velocities after alignment via sparse KMA as shown in the top-center panel have been coloured in blue for boys (39 curves) and pink for girls (54 curves). No pattern is visible, and it looks like both boys and girls might show (or not show) a mid-spurt. Therefore, no gender stratification can be detected in the amplitude clusters. More interestingly, in the bottom-right panel of Figure \ref{fig:growth} the estimated affine warping functions are coloured according to gender consistently to the panel above: this plot evidently shows that the gender stratification is present in the phase. Indeed, most girls show larger slope and intercepts as compared to boys, coherently to the well-known biological knowledge stating that girls grow stochastically earlier and faster than boys. All these findings are completely coherent to what had previously been reported on the same dataset \citep{ssvv,sangalli2010functional}, with the added insight given by the localized amplitude grouping structure. To further confirm such conclusion, we can analyse the estimated grouping structure, and compare it to the one detected in \cite{fv17}: the clusters detected via sparse KMA nearly completely coincide to those found in \cite{fv17}, with only 7 curves out of 93 which are assigned differently. This result is quite striking, since it validates the localized amplitude features detected in \cite{fv17}, but it combines those findings with a coherent decoupling of phase and amplitude variability.

\section{Analysis of the ICA centerlines of the AneuRisk65 dataset}\label{sec:aneurisk}

We illustrate here another possible application of interest for the sparse clustering and alignment method: the AneuRisk65 data set. The AneuRisk project was a large scientific endeavour (see \texttt{https://statistics.mox.polimi.it/aneurisk/}) aimed at studying how the cerebral vessel morphology could impact the pathogenesis of cerebral aneurysms. The AneuRisk65 data set is based on a set of three-dimensional angiographic images taken from 65 subjects, hospitalized at Niguarda Hospital in Milan (Italy) from September 2002 to October 2005, who were suspected of being affected by cerebral aneurysm. From the angiographies of the 65 subjects, among other information, estimates of the ICA centerlines were obtained via three-dimensional free-knot regression splines \citep{sangalli2009efficient}. The first derivatives $x^\prime$, $y^\prime$ and $z^\prime$ of the estimated curves, already displayed and commented in Figure \ref{fig:aneurisk_data}, show an evident phase variability related to the very different dimensions and proportions of the subjects' skulls. However, this variability is ancillar to the scope of the analysis, that is the investigation of a possible grouping structure in the curves' shapes, related to differences in the cerebral morphology across subjects. Therefore, joint clustering and alignment of these curves was already presented in several papers \citep{sangalli2010functional,ssvv}, and two different amplitude clusters were detected and associated to the aneurysm position along the ICA. One of the evident results of those analyses, moreover, was that the shape differences were mostly localized in a specific part of the domain, closer to the brain: hence, we aim here at verifying that the same decoupling of phase and amplitude variability can be detected via sparse clustering and alignment, when also managing to automatically select the part of the domain most relevant to the clustering scope.

\begin{figure}[t]
\centerline{\includegraphics[width=\textwidth]{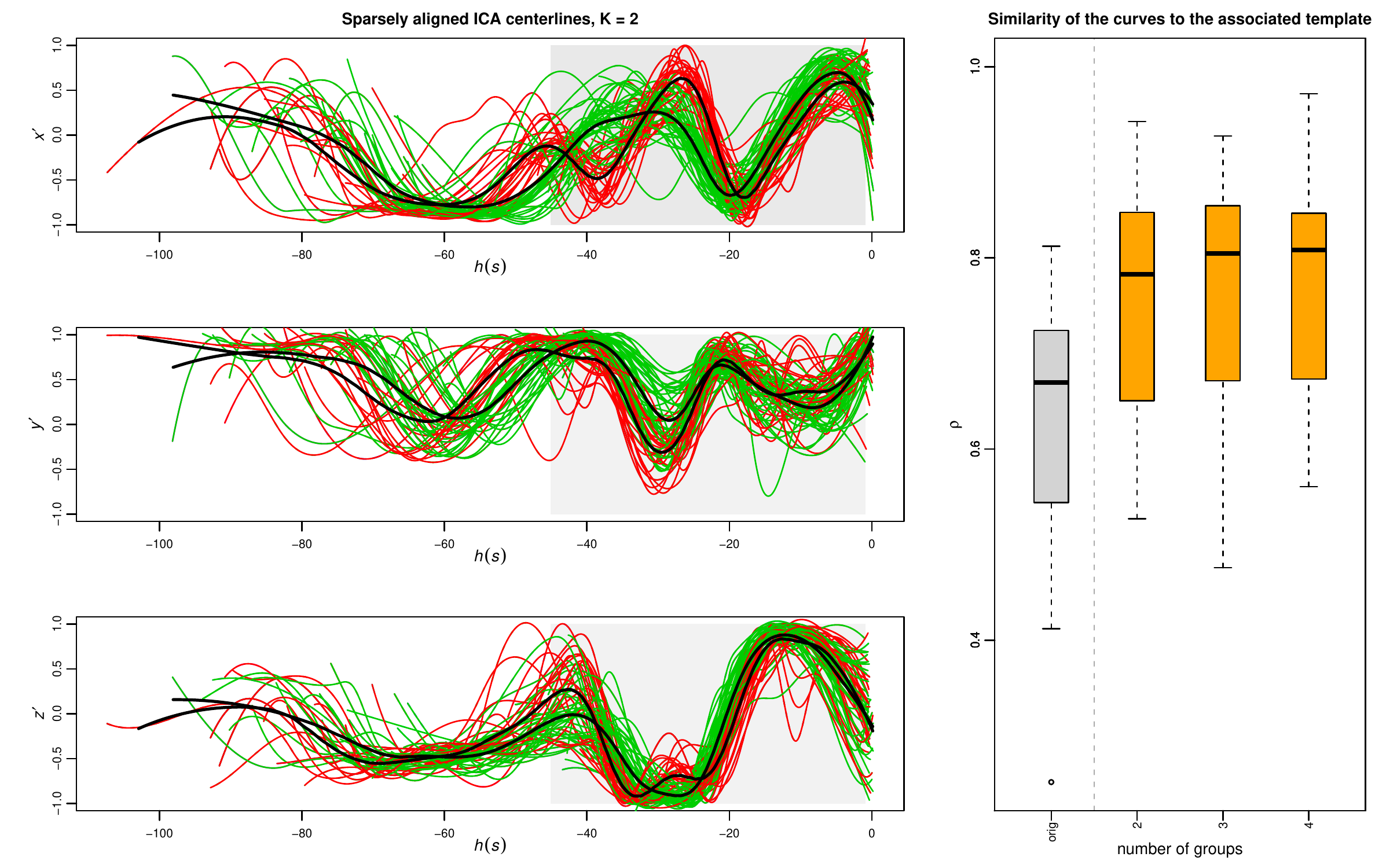}}
\caption{results of sparse KMA on the AneuRisk65 dataset. Left panels, sparsely clustered and aligned first derivatives of the three-dimensional ICA centerlines (from top to bottom, $x^\prime$, $y^\prime$ and $z^\prime$) via sparse KMA with 2 clusters, coloured according to the estimated cluster labels; the grey area in the plot indicates the region of the domain where the weighting function is different from zero. Right panel, value of the functional measure of similarity between each ICA centerline and the associated template, for the original data (shown in Figure \ref{fig:aneurisk_data}, grey boxplot), and for the sparsely aligned data with different $K$s (orange boxplots).\label{fig:aneurisk_results}}
\end{figure}

For this application, we need to adapt the sparse KMA in several ways: first of all, we specify the sparse clustering and alignment variational problem (\ref{eq:sparseClustAlign}) with the functional measure defined in (\ref{eq:rhoKMA}), since this measure based on the curves' derivatives (it is basically a semi-covariance in $H^1$) is more suited for capturing the vessel morphology. This means that, as detailed in (\ref{eq:WCSSsparseClustAlignRho}), the formulation of the variational problem will be based on the maximization of the within-cluster similarity, and not on the maximization of the between-cluster distance. Secondly, the curves are in this case multidimensional, since the statistical unit of the AneuRisk65 data set is a function $f : \mathbb{R}\rightarrow\mathbb{R}^3$. Hence, the natural generalization of the method to multidimensional curves will be used, i.e., the functional similarity in (\ref{eq:rhoKMA}) will be defined as an average of the similarity along the three dimensions. Finally, as it is evident from a look at the curves in Figure \ref{fig:aneurisk_data}, the curves' domains are very different: from the common starting point at the terminal bifurcation of the ICA, some curves are observed for quite long portions of the ICA towards the heart, while others are pretty short, and this uniquely depends on the scan position during the examination (i.e., this has nothing to do with the morphology). Therefore, we need a robust implementation for the estimation of the cluster-specific templates, because on the leftmost portion of the domain very few curves are observed. For this reason, we estimate the template over the domain identified by the union of the domains of the curves assigned to the cluster with the same estimation procedure as defined in \cite{ssvv}, i.e., we use \emph{Loess}, an adaptive fitting that keeps the variance as constant as possible along the abscissa \citep{hastie43tibshirani}.

The sparse KMA procedure is set as follows: we vary $K$ from 2 to 4 (more than 4 groups seems unreasonable given the sample dimension), and select the groups according to an elbow criterion on the final similarity to the template of the associated cluster (same criterion used in \cite{ssvv}). The sparsity parameter $m$ is set to 60\% after some tuning, percentage of allowed warping to 10\% (since misalignment is quite extreme in this application), and tolerance for the iterative procedure to $0.01$, to ensure good convergence in this extreme situation. Figure \ref{fig:aneurisk_results} shows, in the right panel, the boxplot of the initial similarity of the misaligned curves to the overall template (as shown in Figure \ref{fig:aneurisk_data}) in grey, and the final similarity of the aligned curves to the associated template after sparse KMA with $K=2,3,4$ in orange: it is evident that, after an initial strong improvement, increasing the groups to more than 2 does not correspond in any increase of the within-cluster similarity. In the left panels of the same figure the aligned ICA centerlines after sparse 2-mean alignment are shown (from top to bottom, $x^\prime$, $y^\prime$ and $z^\prime$), coloured according to the estimated cluster labels, and with grey background in the portion of the aligned domain where $w(\cdot)$ is non-zero. The resulting cluster morphology and the 2 cluster templates show an astonishing coherence to the results shown in \cite{ssvv} (Figure 12, bottom-right panels): the template shapes are essentially the same, and the remaining within-cluster variability seems comparable. Not only, the grouping structure detected by KMA in \cite{ssvv} is indeed the same as detected here, with only 4 ICA centerlines being differently assigned, all corresponding to a non-aneurysm subject (which means that the reason for the different assignment might indeed be that these ICAs are indeed ``different''). Finally, the portion of the domain selected from $w(\cdot)$ is essentially the same portion of the centerlines which was manually outlined in \cite{ssvv} (Figure 13) as the relevant one. For all these reasons, analysing the AnueRisk65 data set via sparse KMA provides an automatic validation of the previous findings, thus undoubtedly confirming the advantages of a jointly sparse procedure, which can give more insight on the data and more immediate interpretability of the final result without compromising the capability of decoupling phase and amplitude variability.

\section{Conclusions and discussion of further developments}\label{sec:conc}

In this paper, a novel method for jointly performing sparse clustering and alignment of functional data was introduced and described in details, together with a discussion of its relevant mathematical properties. Sparse clustering and alignment was then tested on simulated data, and its performance compared with a competitive approach. Finally, results obtained with this method on two case studies, the well-known growth curves from the Berkeley Growth Study and the AneuRisk65 data set, were discussed in the light of previous findings. Sparse clustering and alignment is a completely novel methodological approach that fills a gap in the FDA literature on joint clustering and alignment. Results seem very promising, and the method showed to work pretty well even in the presence of a quite complex data structure. When compared to previous findings in the real case studies, novel interesting insights on the data were enlightened by the use of joint sparse clustering and alignment.

Nonetheless, the method is still in its infancy. First and foremost, a strategy for automatically tuning the sparsity parameter $m$ would be beneficial for avoiding many trial-and-error attempts. One possibility could be following the GAP statistics permutation-based approach proposed in \cite{fv17}, even though this strategy might need some adaptations to work in case misalignment is also present. Another approach, more targeted to the misalignment level observed in the data, could be computing some sufficient statistics describing phase and amplitude variation, and then relying on these for fixing $m$. Interestingly, we have already observed that increasing $m$ when the sparsity in the data looks more extreme might not be optimal, since misalignment plays a role as well: depending on the observed variation in the phase, it might be beneficial to keep $m$ moderately small, especially since the cluster templates are substantially different only on the non-zero part of $w$. This suggests a possible link between the optimal $m$ for a given dataset, and the level of phase / amplitude variation in the data: this might be a very relevant direction for future research. 

Secondly, we observed in the simulation studies described in Section \ref{sec:simulations} that often sparse clustering and alignment would result in less extreme curve alignment than, for example, clustering and alignment alone. A question might then arise in the reader whether this lack of alignment might be problematic. For what we have observed so far, it looked like there could exist a trade-off between level of sparsity and goodness of alignment: methods more explicitly targeted to alignment would capture phase variability pretty well, but at the cost of a worse estimation of the cluster templates, especially if these differ in very localized areas of the domain. Sparse clustering and alignment would instead manage to estimate the templates very well, even when keeping $m$ pretty small ``to let the data speak'', but it would leave some of the variation in the phase. Undoubtedly, a better understanding of the relationship between $m$ and $W$ would spread some light on this open aspect, as well.

Finally, another point that surely needs further inspection is the possible definition of sparse clustering and alignment also for $K=1$. Currently the method does not handle the case of one group, because the weighting function $w$ is not properly defined. One easy fix, quite coherent from a purely theoretical standpoint, would be stating that sparse clustering and alignment degenerates into pure alignment when $K=1,$ as is the case for clustering and alignment with $K=1$. However, one might then speculate whether the clustering might be present in the phase as well, and maybe in a sparse fashion. Therefore, it could be interesting to think of a generalization of sparse clustering and alignment capable of estimating the grouping structure both in the amplitude and in the phase variability, with the possibility to use a specific sparsity constraint for each of the two clusterings.

Other future research directions include, but are not limited to: investigating the global convergence of the algorithm, and / or finding ways to tackle the variational problem (\ref{eq:sparseClustAlign}) globally; developing a probabilistic approach to sparse clustering and alignment, for example by taking into account the point-wise curves density; and finally performing model-based clustering. All of these would lead to a nice paper in themselves, and are thus left to future speculation.

\section*{Implementation}
The method is implemented in \texttt{R} \citep{R}, and all functions are available upon request. A repository on sparse clustering methods for functional data is under construction at \texttt{https://github.com/valeriavitelli}.


\bibliographystyle{elsarticle-harv} 
\bibliography{myrefs}

%
%
%
%
\end{document}